\documentclass[11pt]{article}
\usepackage{epsf}
\usepackage{color}
\usepackage{epsfig}
\usepackage{times}

\usepackage{graphicx}
\usepackage{amsmath}
\usepackage{latexsym}
\usepackage{amssymb}
\usepackage{amsfonts}
\usepackage{url}

\usepackage{a4}

\newtheorem{theorem}{Theorem}
\newtheorem{remark}{Remark}

\newtheorem{proposition}[theorem]{Proposition}
\newenvironment{proof}[1][Proof]{\textbf{#1.} }
     {\    \rule{0.5em}{0.5em}}

\newcommand{\eps}{\varepsilon}

\newcommand{\ad}{{\rm ad}}

\newcommand{\e}{\ensuremath{\mathrm{e}}}

\newcommand{\R}{\mathbb{R}}

\newcommand{\T}{\mathbb{T}}
\newcommand{\ZZ}{\mathbb{Z}}
\newcommand{\NN}{\mathbb{N}}
\newcommand{\LL}{<}
\newcommand{\RR}{>}
\newcommand{\pa}{\partial}

\begin{document}

\title{Continuous changes of variables and the Magnus expansion} 

\author{ Fernando Casas\footnote{Universitat Jaume I, IMAC, Departament de Matem\`atiques, 12071~Castell\'on, Spain. Email: fernando.casas@uji.es}, \and 
Philippe Chartier\footnote{INRIA, Universit\'e de Rennes 1, Campus de Beaulieu, 35042 Rennes, France. Email: 
philippe.chartier@inria.fr}, \and
Ander Murua\footnote{Konputazio Zientziak eta A.A. Saila, Infomatika Fakultatea, UPV/EHU, 20018 
Donostia-San Sebasti\'an, Spain. Email: ander.murua@ehu.es}}

\maketitle

\begin{abstract}

In this paper, we are concerned with a formulation of Magnus and Floquet-Magnus expansions for general nonlinear differential equations. To this aim, we introduce suitable 
continuous variable transformations generated by operators. As an application of the simple formulas so-obtained, we explicitly compute the first terms of the Floquet-Magnus expansion for the Van der Pol oscillator and the nonlinear
Schr\"odinger equation on the torus. 
\end{abstract}

\section{Introduction}

The Magnus expansion constitutes nowadays a standard tool for obtaining both analytic and numerical approximations to the
solutions of non-autonomous linear differential equations. In its simplest formulation, the Magnus expansion \cite{magnus54ote} aims to construct
the solution of the linear differential equation
\begin{equation}  \label{de.1}
   Y'(t) = A(t) Y(t), \qquad Y(0) = I,
\end{equation}
where $A(t)$ is a $n \times n$ matrix, as
\begin{equation}  \label{me.1}
   Y(t) = \exp \Omega(t),
\end{equation}      
where $\Omega$ is an infinite series
\begin{equation}   \label{me.2}
   \Omega(t) = \sum_{k=1}^{\infty} \Omega_k(t), \qquad \mbox{ with } \qquad \Omega_k(0) = 0,
\end{equation}
whose terms are increasingly complex expressions involving iterated integrals of nested commutators of the matrix $A$
evaluated at different times.  

Since the 1960s the Magnus expansion (often with different names)
has been used in many different fields, ranging from nuclear, atomic and molecular physics 
to nuclear magnetic
resonance and quantum electrodynamics, mainly in connection with perturbation theory. 
More recently, it has also been the starting point to construct numerical
integration methods  in the realm of geometric numerical integration
(see \cite{blanes09tme} for a review), when preserving  the main qualitative features of the exact solution, such as its invariant quantities or the geometric structure is at issue
\cite{blanes16aci,hairer06gni}. The convergence
of the expansion is also an important feature and several general results are available \cite{blanes98maf,casas07scf,moan08cot,lakos17cef}. 

Given the favourable properties exhibited by the Magnus expansion in the treatment of the linear problem (\ref{de.1}), it comes as 
no surprise that several generalizations have been proposed along the years. We can mention, in particular, equation (\ref{de.1})
when the (in general complex) matrix-valued function $A(t)$ is periodic with period $T$. In that case, it is possible to 
combine the Magnus expansion with the Floquet theorem \cite{coddington55tod} and construct the solution as
\begin{equation}  \label{flo.1}
  Y(t) = \exp(\Lambda(t)) \, \exp(t F),
\end{equation}  
where $\Lambda(t+T) = \Lambda(t)$ and both $\Lambda(t)$ and $F$ are series expansions
\begin{equation}   \label{flo.2}
   \Lambda(t) = \sum_{k=1}^{\infty} \Lambda_k(t), \qquad F = \sum_{k=1}^{\infty} F_k,
\end{equation}
with $\Lambda_k(0) = 0$ for all $k$. This is the so-called Floquet--Magnus expansion \cite{casas01fte}, and has been widely used in problems of
solid state physics and nuclear magnetic resonance \cite{kuwahara16fmt,mananga16otf}. 
Notice that, due to the periodicity of $\Lambda_k$, the constant
term $F_n$ can be independently obtained as $F_k = \Omega_k(T)/T$ for all $k$.

In the general case of a nonlinear ordinary differential equation in $\mathbb{R}^n$, 
\begin{equation}  \label{fnl.1} 
  x' = g(x,t), \qquad x(0) = x_0 \in \mathbb{R}^n,
\end{equation}
the usual procedure to construct the Magnus expansion requires first to
transform  (\ref{fnl.1}) into a certain linear  equation  involving operators \cite{agrachev79ter}. This
is done 
by introducing the Lie derivative associated with $g$ and the family of linear 
transformations $\Phi_t$ such that $\Phi_t [f] = f \circ \varphi_t$, where $\varphi_t$ denotes the exact flow
defined by (\ref{fnl.1}) and $f$ is any (infinitely) differentiable map $f : \mathbb{R}^n \longrightarrow \mathbb{R}$. The operator
$\Phi_t$ obeys a linear differential equation which is then formally solved with the corresponding Magnus expansion \cite{blanes09tme}.
Once the series is truncated, it corresponds to the Lie derivative of some function $W(x,t)$. 
Finally, the solution at some given time $t=T$ can be approximated by determining the 
1-flow of the autonomous differential equation
\[
  y' = W(y,T), \qquad y(0) = x_0
\]       
since, by construction, $y(1)   \simeq \varphi_T(x_0)$. Clearly, the whole procedure is different and more involved
than in the linear case. It is the purpose of this work to provide a unified framework to derive the Magnus expansion in a simpler
way without requiring the apparatus of chronological calculus. This will be possible by applying the continuous
transformation theory developed by Dewar in perturbation theory in classical mechanics \cite{dewar76rcp}. In that context, the Magnus
series is just the generator of the continuous transformation sending the original system (\ref{fnl.1}) to the trivial one $X' = 0$.
Moreover, the same idea can be applied to the Floquet--Magnus expansion, thus establishing a natural connection with the
stroboscopic averaging formalism. In the process, the relation with pre-Lie algebras and other combinatorial objects will appear
in a natural way. 

The plan of the paper is as follows. We review several procedures to derive the Magnus expansion for the linear equation
(\ref{de.1}) in section \ref{sec.2} and introduce a binary operator that will play an important role in the sequel. In section \ref{sec.3}
we consider continuous changes of variables and their generators in the context of general ordinary differential equations, whereas in sections \ref{sec.4} and \ref{sec.5} we apply this formalism for constructing the Magnus and Floquet--Magnus
expansions, respectively, in the general nonlinear setting. There, we also show how they reproduce the classical expansions
for linear differential equations. As a result, both expansions can be considered as the output of appropriately continuous changes
of variables rendering the original system into a simpler form. {Finally, in section \ref{sec.6} we illustrate the techniques
developed here by considering two examples: the Van der Pol oscillator and the nonlinear Schr\"odinger equation with
periodic boundary conditions.}

\section{The Magnus expansion for linear systems}
\label{sec.2}

There are many ways to get the terms of the Magnus series (\ref{me.2}). If we introduce
a (dummy) parameter $\varepsilon$ in eq. (\ref{de.1}), i.e., we replace $A$ by $\varepsilon A$, then the successive terms in
\begin{equation}  \label{mls1}
  \Omega(t) = \varepsilon \Omega_1(t) + \varepsilon^2 \Omega_2(t) + \varepsilon^3 \Omega_3(t) + \cdots
\end{equation}
can be determined by inserting $\Omega(t)$  into eq. (\ref{de.1}) and computing the derivative of the matrix exponential, thus
arriving at  \cite{iserles00lgm}
\begin{equation}  \label{mls2}
  \varepsilon A = d \exp_{\Omega} (\Omega') \equiv \sum_{k=0}^{\infty} \frac{1}{(k+1)!} \mathrm{ad}_{\Omega}^k (\Omega') =
  \Omega' + \frac{1}{2} [ \Omega, \Omega'] + \frac{1}{3!} [\Omega, [\Omega, \Omega']] + \cdots
\end{equation}
where $[A,B]$ denotes the usual Lie bracket (commutator) and $\mathrm{ad}_A^j B = [A, \mathrm{ad}_A^{j-1} B]$,
$\mathrm{ad}_A^0 B = 0$. At this point
it is useful to introduce the linear operator
\begin{equation}  \label{preL}
  H(t) = (F \rhd G) (t):= \int_0^t [F(u), G(t)] du
\end{equation}
so that, in terms of 
\begin{equation}  \label{erre}
   R(t) = \frac{d }{dt} \Omega(t) = \varepsilon R_1(t) + \eps^2 R_2(t) + \eps^3 R_3(t)+\cdots,
\end{equation}  
equation (\ref{mls2}) can be written as 
\begin{equation}   \label{eq:Req}
 \varepsilon\,  A = R + \frac12 \, R \rhd R
+ \frac{1}{3!} \, R \rhd R \rhd R
+ \frac{1}{4!} \, R \rhd R \rhd R \rhd R+ \cdots.
\end{equation}
Here we have used the notation
\[
  F_1 \rhd F_2 \rhd \cdots \rhd F_m = F_1 \rhd (F_2 \rhd \cdots \rhd F_m).
\]   
Now the successive terms $R_j(t)$ can be determined by substitution of (\ref{erre}) into (\ref{eq:Req}) 
and comparing like powers of $\eps$. In particular, this gives
\begin{eqnarray*}
  R_1 &=& A, \\
  R_2 &=& -\frac{1}{2}\, R_1 \rhd R_1 \ = \ -\frac{1}{2}\, A \rhd A, \\
  R_3 &=&  -\frac{1}{2}\, ( R_2 \rhd R_1 + R_1 \rhd R_2) - \frac{1}{6}\, R_1 \rhd R_1 \rhd R_1 \\
&=&
 \frac{1}{4}\, (A \rhd A) \rhd A + \frac{1}{12}\, A \rhd A \rhd A
\end{eqnarray*}
Of course, equation (\ref{mls2}) can be inverted, thus resulting in
\begin{equation} \label{inv}
  \Omega' = \sum_{k=0}^{\infty}  \frac{B_k}{k!} \ad_{\Omega}^k (\eps A(t)),  \qquad \Omega(0) = 0
\end{equation}
where the $B_j$ are the Bernoulli numbers, that is
\begin{eqnarray*}
  \frac{x}{e^x-1} &=& 1 + B_1 x + \frac{B_2}{2!} x^2 + \frac{B_4}{4!} x^4
+ \frac{B_6}{6!} x^6+ \cdots\\
&=& 1 - \frac12 x + \frac{1}{12} x^2 - \frac{1}{720} x^4 + \cdots  
\end{eqnarray*}
In terms of $R$, equation (\ref{inv}) can be written as
\begin{equation}   \label{eq.12b}
\eps^{-1}\, R = A - B_1 \, R \rhd A 
+ \frac{B_2}{2!} \, R \rhd R \rhd A
+ \frac{B_4}{4!} \, R \rhd R \rhd R \rhd R \rhd A + \cdots.
\end{equation}
Substituting (\ref{erre}) in eq. (\ref{eq.12b}) and working out the resulting expression, one arrives to the following 
recursive procedure allowing to determine
the successive terms $R_j(t)$:
\begin{eqnarray}    \label{recur1}
  &  &  S_m^{(1)} = [\Omega_{m-1}, A], \qquad S_m^{(j)} = \sum_{n=1}^{m-j} [ \Omega_n, S_{m-n}^{(j-1)}], \quad
    2 \le j \le m-1  \nonumber  \\
    & & R_1(t) =A(t), \qquad \qquad
     R_m(t) = \sum_{j=1}^{m-1} \frac{B_j}{j!}  S_m^{(j)}(t), \quad m \ge 2. 
\end{eqnarray}    
Notice in particular that 
\[
   S_2^{(1)} = A \rhd A, \qquad S_3^{(1)} = -\frac{1}{2} (A \rhd A) \rhd A, \qquad S_3^{(2)} = A \rhd A \rhd A.
\]   
At this point it is worth remarking that any of the above procedures can be used to write each $R_j$ in terms of the
binary operation $\rhd$ and the original time-dependent linear operator $A$, which gives in general one term per binary tree, as in~\cite{iserles99ots,iserles00lgm}, or equivalently, one term per planar rooted tree. However,
the binary operator $\rhd$ satisfies, as a consequence of the Jacobi identity of the Lie bracket of vector fields and the integration by parts formula, the so-called pre-Lie relation
\begin{equation}
  \label{eq:preLie}
  F \rhd G \rhd H - (F \rhd G) \rhd H
= G \rhd F \rhd H- (G \rhd F) \rhd H,
\end{equation} 
  As shown in~\cite{ebrahimi09ama}, this relation can be used to rewrite each $R_j$ as a sum of fewer terms, the number of terms being less than or equal to the number of rooted trees with $j$ vertices. For instance, the formula for $R_4$ can be written in the simplified form
\[
R_4 = -\frac{1}{6}\, ((A \rhd A) \rhd A) \rhd A - \frac{1}{12}\,
 A \rhd (A \rhd A) \rhd A
\]
upon using the pre-Lie relation (\ref{eq:preLie}) for $F=G\rhd G$ and $H=G$.

If, on the other hand, one is more interested in getting an explicit expression for $\Omega_j(t)$, 
ithe usual starting point is to express the solution of (\ref{de.1}) as the series
 \begin{equation}   \label{N.1}
    Y(t) = I + \sum_{n=1}^{\infty} \int_{\Delta_n(t)} A(t_1) A(t_2) \cdots A(t_n) \, dt_1 \cdots dt_n,
 \end{equation}
 where 
 \begin{equation}  \label{N.2b}
   \Delta_n(t) = \{ (t_1, \ldots, t_n) : \, 0 \le t_n \le \cdots \le t_1 \le t \}
\end{equation}
and then compute formally the logarithm of (\ref{N.1}). Then one gets
\cite{bialynicki69eso,mielnik70cat,strichartz87tcb,agrachev94tsp}
\[
   \Omega(t)=  \log Y(t) = \sum_{n=1}^{\infty}\Omega_{n}(t),
\]
with
\begin{equation}   \label{me.ite3}
  \Omega_n(t) =   \frac{1}{n} \sum_{\sigma \in S_n} \, (-1)^{d_{\sigma}} \, \frac{1}{ \binom{n-1}{d_{\sigma}}} \, 
   \int_{\Delta_n(t)} \, A(t_{\sigma(1)}) A(t_{\sigma(2)}) \cdots A(t_{\sigma(n)}) \, dt_1 \cdots dt_n.
\end{equation}
Here $\sigma \in S_n$ denotes a permutation of $\{ 1, 2, \ldots, n \}$. An expression in terms only of independent
commutators can be obtained by using the class of bases proposed by  Dragt \& Forest \cite{dragt83con} for the
Lie algebra generated by
the operators $A(t_1), \ldots A(t_n)$, thus resulting in \cite{arnal18agf}
\begin{equation}   \label{me.com3}
\begin{aligned}
  & \Omega_n(t) = \frac{1}{n} \sum_{\sigma \in S_{n-1}} \, (-1)^{d_{\sigma}} \frac{1}{\binom{n-1}{d_{\sigma}}} \, 
    \int_0^t dt_1 \int_0^{t_1} dt_2 \cdots \int_0^{t_{n-1}} dt_n \,  \\
    &  \qquad\qquad\qquad  [A(t_{\sigma(1)}), [A(t_{\sigma(2)}) \cdots 
    [A(t_{\sigma(n-1)}), A(t_n)] \cdots ]],
\end{aligned}    
\end{equation}
where now $\sigma$ is a permutation of $\{1, 2, \ldots, n-1 \}$ and $d_{\sigma}$ corresponds to the number descents of $\sigma$. 
We recall that $\sigma$ has a descent in $i$ if $\sigma(i) > \sigma(i+1)$, $i=1,\ldots,n-2$.

\section{Continuous changes of variables}
\label{sec.3}

Our purpose in the sequel is to generalize the previous expansion 
to general nonlinear differential equations. It turns out that a suitable tool for that purpose is the use of 
continuous variable transformations generated by operators \cite{dewar76rcp,cary81ltp}. We therefore
summarize next its main features.

Given a generic ODE system of the form
\begin{equation}
\label{eq:odex}
\frac{d}{dt} x= f(x,t),
\end{equation}
the idea is to apply some near-to-identity change of variables $x \longmapsto X$ that transforms the original system (\ref{eq:odex}) into
\begin{equation}
\label{eq:odeX}
\frac{d}{dt} X= F(X,t),
\end{equation}
where the vector field $F(X,t)$ adopts some desirable form. 
In order to do that in a convenient way, we apply a one-parameter family of time-dependent transformations of the form
\begin{equation*}
z = \Psi_{s}(X,t), \qquad s \in \R,
\end{equation*}
such that $\Psi_{0}(X,t) \equiv X$, and $x=\Psi_{1}(X,t)$ is the change of variables that we seek. In this way, one continuously varies $s$ from $s=0$ to $s=1$ to move from the trivial change of variables $x=X$ to $x=\Psi_{1}(X,t)$, so that for each solution $X(t)$ of (\ref{eq:odeX}), the function $z(t,s)$ defined by  $z(t,s)=\Psi_{s}(X(t),t)$ satisfies a differential equation
\begin{equation}
\label{eq:odeV}
\frac{\partial}{\partial t} z= V(z,t,s).
\end{equation}
In particular, we will have that $F(X,t)=V(X,t,0)$ and $f(x,t)=V(x,t,1)$. 

Next, the near-to-identity family of maps $X\longmapsto z=\Psi_{s}(X,t)$ is defined 
in terms of  a differential equation in the independent variable $s$, 
\begin{equation}
\label{eq:odeW}
\frac{\partial}{\partial s} z(t,s) = W(z(t,s),t,s)
\end{equation}
 by requiring that $z(t,s)=\Psi_{s}(z(t,0),t)$ for any solution $z(t,s)$ of (\ref{eq:odeW}). The map $\Psi_s(\cdot,t)$ will be near-to-identity if $W(z,t,s)$ is of the form 
 \begin{equation*}
W(z,t,s) = \eps W_1(z,t,s) + \eps^2 W_2(z,t,s) + \cdots,
 \end{equation*}
for  some small parameter $\eps$.

\begin{proposition}[\cite{dewar76rcp}]
\label{prop:1}
Given $F$ and $W=\eps W_1 + \eps^2 W_2 + \cdots$, the right-hand side $V$ of the continuously transformed system (\ref{eq:odeV}) can be uniquely determined (as a formal series in powers of $\eps$) from $V(X,t,0)=F(X,t)$ 
and
 \begin{equation}
 \label{eq:VW}
\frac{\partial}{\partial s} V(x,t,s) -\frac{\partial}{\partial t} W(x,t,s) = W'(x,t,s) V(x,t,s)-V'(x,t,s) W(x,t,s),
\end{equation}
where $W'$ and $V'$ refer to the differentials $\partial_x W$ and $\partial_x V$, respectively.
\end{proposition}
\begin{proof}
By partial differentiation of both sides in (\ref{eq:odeW}) with respect to $t$ and partial differentiation of both sides in (\ref{eq:odeV}) with respect to $s$, we conclude that (\ref{eq:VW}) holds for all $x=z(s,t)=\Psi_s(x_0,t)$ with arbitrary $x_0$ and all $(t,s)$. 
One can show that the equality (\ref{eq:VW}) holds for arbitrary $(x,t,s)$ by taking into account that, for given $t$ and $s$,   $x_0 \mapsto x=\Psi_s(x_0,t)$ is one-to-one.

Now,  since $V(x,t,0)=F(x,t)$, we have that 
\begin{equation}
\label{eq:VFS}
V(x,t,s) = F(x,t) + \int_{0}^{s} S(x,t,\sigma) \, d\sigma,
\end{equation}
where $S = \frac{\partial}{\partial t} W + W' V -V' W$. Clearly, the successive terms of 
$$V=F + \eps V_1 + \eps^2 V_2 + \cdots$$
 are uniquely determined by equating like powers of $\varepsilon$ in (\ref{eq:VFS}).
\end{proof}
%

In the sequel we always assume that  the generator $W$ of the change of variables 
\begin{itemize}
\item[(i)] does not depend on $s$, and
\item[(ii)] $W(x,0,s) \equiv 0$, so that $\Psi_s(x,0)=x$ and   $x(0)=X(0)$.
\end{itemize}
 The  successive terms in the expansion of $V(x,t,s)$ in Proposition~\ref{prop:1} can be conveniently computed with the help of a binary operation
$\rhd$ on  maps $\R^{d+1} \longrightarrow \R^d$ defined as follows. Given two such maps $P$ and $Q$, then 
$P \rhd Q$ is a new map whose evaluation at $(x,t) \in \R^{d+1}$ takes the value
\begin{equation}
\label{eq:rhd}
(P \rhd  Q)(x,t) = \int_{0}^{t} ( P'(x,\tau) Q(x,t)-  Q'(x,t) P(x,\tau)) d\tau.
\end{equation}
Under these conditions, from Proposition~\ref{prop:1}, we have that 
\begin{equation}  \label{eq:rhd2}
\frac{\partial}{\partial s} V(x,t,s) -\frac{\partial}{\partial t} W(x,t) = \big[ W(x,t), V(x,t,s) \big]
\end{equation}
with the notation
\begin{equation}   \label{bracket}
 \Big[ W(x,t), V(x,t,s) \Big] := W'(x,t) V(x,t,s)-V'(x,t,s) W(x,t)
\end{equation}
for the Lie bracket.

Equation (\ref{eq:rhd2}), in terms of 
\begin{equation}
\label{eq:R}
R(x,t):= \frac{\partial}{\partial t} W(x,t),
\end{equation}
reads
\[
\frac{\partial}{\partial s} V(x,t,s) =  R(x,t) + \int_0^t \left( R'(x,\tau) V(x,t,s)-V'(x,t,s) R(x,\tau) \right) d \tau
\]
or equivalently
\begin{equation*}
V_s = V_0 + s R + R \rhd \left(\int_{0}^{s} V_{\sigma} d \sigma\right),
\end{equation*}
where we have used the notation $V_s(x,t):=V(x,t,s)$.  Since $V(X,t,0) = F(X,t)$, then
\begin{align}
  \label{eq:V}
 V(\cdot,\cdot,s) &= s\, R + \frac{s^2}{2} \, R \rhd R 
+ \frac{s^3}{3!} \, R \rhd R \rhd R + \cdots \\
& \nonumber 
+ F + s\, R \rhd F + \frac{s^2}{2}  \, R  \rhd R \rhd F
+ \frac{s^3}{3!} \, R \rhd R \rhd R \rhd F
 + \cdots
\end{align}
with the convention $F_1 \rhd F_2 \rhd \cdots \rhd F_m = F_1 \rhd (F_2 \rhd \cdots \rhd F_m)$.

We thus have the following result:
\begin{proposition}
\label{prop:2}
A  change of variables $x=\Psi_1(X,t)$ defined  in terms of  a continuous change of variables $X \longmapsto z = \Psi_s(X,t)$ with generator
\begin{equation}
\label{eq:W}
W(x,t)=\eps \, W_1(x,t) + \eps^2 \, W_2(x,t) + \cdots
\end{equation}
and $W(x,0) \equiv x$, transforms the system of equations (\ref{eq:odex}) into (\ref{eq:odeX}), where $f$ and $F$ are related by
\begin{align}
  \label{eq:FRf}
f &=  R + \frac{1}{2} \, R \rhd R 
+ \frac{1}{3!} \, R \rhd R \rhd R + \frac{1}{4!} \, R \rhd R \rhd R \rhd R+ \cdots \\
& \nonumber 
+ F + R \rhd F + \frac{1}{2}  \, R  \rhd R \rhd F
+ \frac{1}{3!} \, R \rhd R \rhd R \rhd F
 + \cdots
\end{align}
and $R$ is given by (\ref{eq:R}).
\end{proposition}

Proposition~\ref{prop:2} deals with changes of variables such that $X = \Psi_1(X,0)$ (as a consequence of $W(X,0) \equiv X$), so that the initial value problem obtained by supplementing (\ref{eq:odex}) with the initial condition $x(0) = x_0$ is transformed into (\ref{eq:odeX}) supplemented with $X(0)=x_0$. 

More generally, one may consider generators $W(\cdot,t)$ within some class $\mathcal{C}$ of time-dependent smooth vector fields such
that the operator $\partial_t: \mathcal{C} \to \mathcal{C}$ is invertible.  
Next result reduces to Proposition~\ref{prop:2}, when one considers some class $\mathcal{C}$ of generators $W(\cdot,t)$ such that $W(x,0) \equiv 0$, so that $\partial_t: \mathcal{C} \to \mathcal{C}$ is invertible, with inverse defined as $\partial_t^{-1} W (x,t) = \int_0^t W(x,\tau)\, d\tau$.

\begin{proposition}
\label{prop:3}
A  change of variables $x=\Psi_1(X,t)$ defined  in terms of  a continuous change of variables $X \longmapsto z = \Psi_s(X,t)$ with generator
\begin{equation}
\label{eq:W1}
W(x,t)=\eps \, W_1(x,t) + \eps^2 \, W_2(x,t) + \cdots
\end{equation}
within some class $\mathcal{C}$ of time dependent smooth vector fields with invertible $\partial_t:\mathcal{C} \to \mathcal{C}$
transforms the initial value problem
\begin{equation}
\label{eq:ivpx}
\frac{d}{dt} x= f(x,t), \quad x(0) = x_0
\end{equation}
into 
\begin{equation}
\label{eq:ivpX}
\frac{d}{dt} X= F(X,t), \quad X(0) = \Psi_1^{-1}(x_0),
\end{equation}
where $f$, $F$, and $R = \partial_t W$ are related by (\ref{eq:FRf}), and the binary operator $\rhd: \mathcal{C} \times \mathcal{C} \to \mathcal{C}$ is defined as 
 \begin{equation}
\label{eq:rhd2b}
P \rhd  Q =  ( \partial_t^{-1}  P')Q  -  Q' ( \partial_t^{-1}  P)  =   [ \partial_t^{-1}  P,  \, Q].
\end{equation}
\end{proposition}
Notice that the operation $\rhd$ of (\ref{eq:rhd2b}) satisfies the pre-Lie relation (\ref{eq:preLie}), and that this
 proposition applies, in particular, to the class
$\mathcal{C}$ of smooth $(2\pi)$-periodic vector fields in $\mathbb{R}^d$ with vanishing average. In that case the operator $\partial_t$ is invertible, with inverse given by
\begin{equation*}
\partial_t^{-1} W(x,t) = \sum_{k \in \ZZ \atop k \ne 0} \frac{1}{ i\, k}\, \e^{i\,  k\,  t} \, \hat W_k(x), \quad \mbox{if}
 \quad W(x,t) = \sum_{k \in \ZZ \atop k \ne 0} \e^{i\,  k\,  t} \, \hat W_k(x).
 \end{equation*}

\section{Continuous transformations and the Magnus expansion}
\label{sec.4}

Consider now an initial value problem of the form
\begin{equation}
\label{eq:odeg}
\frac{d}{dt} x= \eps \, g(x,t), \quad x(0)=x_0,
\end{equation}
where the parameter $\varepsilon$ has been introduced for convenience. As stated in the introduction,
the solution $x(t)$ of this problem (\ref{eq:odex}) can be approximated at a given $t$  as the solution $y(s)$ at $s=1$ of the autonomous initial value problem
\begin{equation*}
\frac{d}{ds} y= \eps \, W_1(y,t):= \eps \int_{0}^{t} g(z,\tau) \, d\tau, \quad y(0)=x_0.
\end{equation*}
This is nothing but the first term in the Magnus approximation of $x(t)$. As a matter of fact, the Magnus expansion is a formal series (\ref{eq:W})
such that, for each fixed value of $t$,  formally $x(t) = y(1)$, where $y(s)$ is the solution of 
\begin{equation*}
\frac{d}{ds} y= W(y,t), \quad y(0)=x_0.
\end{equation*}
The Magnus expansion (\ref{eq:W}) can then be obtained by 
applying a change of variables $x=\Psi_1(X,t)$, defined in terms of  a continuous transformation 
$X \longmapsto z = \Psi_s(X,t)$ with generator $W=W(x,t)$ independent of $s$, that transforms (\ref{eq:odeg}) into 
\begin{equation*}
\frac{d}{dt} X = 0.
\end{equation*}
This can be achieved by 
applying Proposition~\ref{prop:2}  with $F(X,t)\equiv 0$ and $f(x,t)=\eps g(x,t)$, i.e., solving 
\begin{equation}
  \label{eq:M1}
\eps \, g =  R + \frac{1}{2} \, R \rhd R 
+ \frac{1}{3!} \, R \rhd R \rhd R + \frac{1}{4!} \, R \rhd R \rhd R \rhd R+ \cdots 
\end{equation}
for
\begin{equation*}
R(x,t) = \eps\, R_1(x,t) + \eps^2\, R_2(x,t) + \eps^3\, R_3(x,t) + \cdots
\end{equation*}
and determining the generator $W$ as
\begin{equation}
\label{eq:WR}
W(x,t) = \int_0^t R(x,\tau)\, d\tau.
\end{equation}

At this point it is worth analyzing how the usual Magnus expansion for linear systems developed in section \ref{sec.2}
is reproduced with this formalism. To do that, we introduce operators $\Omega(t)$ and $B_s(t)$ such that 
\[
    W(x,t):= \Omega(t) x, \quad V(x,t,s):= B_s(t) x, \qquad \frac{\partial }{\partial t} W(x,t) := R(t) x.
\]
Now equation (\ref{eq:rhd2}) reads
\[
 \left( \frac{\partial }{\partial s} B_s \right) x -   R x = \Omega B_s x - B_s \Omega x
\]
or equivalently  
\begin{equation*}
B_s = B_0 + s R + R \rhd \left(\int_{0}^{s} B_{\sigma} d \sigma\right),
\end{equation*}
where the binary operation $\rhd$ defined in (\ref{eq:rhd}) reproduces (\ref{preL}). Since $B_1(t) = \varepsilon \, A(t)$ and $B_0 = 0$, then (\ref{eq:M1})
is precisely (\ref{eq:Req}). The continuous change of variables is then given by
\[
   X \longmapsto z = \Psi_s(X,t) = \exp \big( s \Omega(t) \big) X
\]
so that 
\[
   x(t) = \Psi_1(X,t) = \e^{\Omega(t)} X(t) = \e^{\Omega(t)} X(0) = \e^{\Omega(t)} x(0)
\]
reproduces the Magnus expansion in the linear case.
In consequence, the expression for each term $W_j(x,t)$ in the Magnus series for the ODE
(\ref{eq:odeg}) can be obtained from the corresponding formula for the linear case with the binary operation (\ref{eq:rhd}) and 
all results collected in section \ref{sec.2} are still valid in the general setting by replacing the commutator by the Lie bracket
(\ref{bracket}).

\section{Continuous transformations and the Floquet--Magnus expansion}
\label{sec.5}

The procedure of section \ref{sec.3} can be extended  when there is a periodic time dependence in the differential equation. In that
case one gets a generalized Floquet--Magnus expansion with agrees with the usual one when the problem is linear. 

As usual, the starting point is the initial value problem
\begin{equation}
\label{eq:odeT}
\frac{d}{dt} x= \eps \, g(x,t), \quad x(0)=x_0,
\end{equation}
where now $g(x,t)$ depends periodically on $t$, with period $T$. As before, we 
apply  a change of variables $x=\Psi_1(X,t)$, defined in terms of  a continuous transformation $X \longmapsto z = \Psi_s(X,t)$ with generator $W=W(x,t)$ that removes the time dependence, i.e.,
that transforms (\ref{eq:odeT}) into 
\begin{equation}
\label{eq:G}
\begin{split}
\frac{d}{dt} X &= \eps\, G(X; \eps) =  \eps\, G_1(X) + \eps^2\, G_2(X) + \eps^3\, G_3(X) + \cdots, \\
X(0) &=x_0.
\end{split}
\end{equation}
{
In addition, the generator $W$ is chosen to be independent of $s$ and periodic in $t$ with the same period $T$.
}

This can be achieved by 
considering Proposition~\ref{prop:2}  with $F(X,t):= \eps\, G(X;\eps)$ and $f(x,t):=\eps \, g(x,t)$,  solving (\ref{eq:FRf}) for the series
\begin{equation*}
R(x,t) = \eps\, R_1(x,t) + \eps^2\, R_2(x,t) + \eps^3\, R_3(x,t) + \cdots
\end{equation*}
Thus, for the first terms one gets 
\begin{align*}
  R_1  & = g  - G_1 \\
  R_2  & = - \frac{1}{2} R_1 \rhd R_1 - R_1 \rhd G_1  - G_2\\
   R_3  & = - \frac{1}{2} (R_1 \rhd R_2 + R_2 \rhd R_1) + \frac{1}{3!} R_1 \rhd R_1 \rhd R_1 \\
     & \quad  -  R_1 \rhd G_2  - R_2 \rhd G_1 - \frac{1}{2} R_1 \rhd R_1 \rhd G_1 - G_3
\end{align*}   
and, in general,
\begin{equation}  \label{eq.recu}
    R_j = U_j - G_j,  \qquad j =1,2,\ldots,
\end{equation}
where $U_j$ only contains terms involving $g$ or the vector fields $U_m$ and $G_m$ of a lower order, i.e., with $m < j$. This
allows one to solve (\ref{eq.recu}) recursively by taking {the average of $U_j$ over one period $T$, i.e.,}   
\[
   G_j(X) = \left\LL U_j(X,\cdot)  \right\RR \equiv  \frac{1}{T} \int_0^T U_j(X,t) dt,
\]
thus ensuring that $R_j$ is periodic. Finally, once $G$ and $R$ are determined, $W$ is obtained from (\ref{eq:WR}),
which in turn determines the change of variables.    

If we limit ourselves to the linear case, $g(x,t) = A(t) x$, with $A(t +T) = A(t)$, then, by introducing the operators
\[
    W(x,t):= \Lambda(t) x, \quad V(x,t,s):= B_s(t) x, \quad G(X):= F X,
\]
the relevant equation is now
\begin{equation*}
B_s = F + s \Lambda' + \Lambda' \rhd \left(\int_{0}^{s} B_{\sigma} d \sigma\right),
\end{equation*}
which, with the additional constraint $B_1(t) = A(t)$, leads to
\begin{align*}
  \Lambda_1'  & = A  - F_1 \\
  \Lambda_2'  & = - \frac{1}{2} \Lambda_1'  \rhd \Lambda_1' - \Lambda_1' \rhd F_1  - F_2
 \end{align*}   
and so on, i.e., exactly the same expressions obtained in \cite{casas01fte}. The transformation is now
\[
   x(t) = \Psi_1(X,t) = \e^{\Lambda(t)} X(t) = \e^{\Lambda(t)} \e^{t F} x(0) 
\]
thus reproducing the Floquet--Magnus expansion in the periodic case \cite{casas01fte}.

Several remarks are in order at this point:
\begin{enumerate}
 \item This procedure has close similarities with several averaging techniques. As a matter of fact, in the quasi-periodic case, it is
equivalent to the high order quasi-stroboscopic averaging treatment carried out in \cite{chartier12afs}.
\item  A different style of high order averaging (that can be more convenient in some practical applications) can be performed by dropping the condition that $W(x,0)\equiv x$, and requiring instead, for instance, that $W(x,t)$ has vanishing average in $t$.
 In that case, the initial condition in (\ref{eq:G}) must be replaced by $X(0) = \Psi_1^{-1}(x_0)$. 
  The generator $W(x,t)$ of the change of variables and the averaged vector fields $\eps\, G(x,t)$ can be similarly computed by considering Proposition~\ref{prop:3} with the class $\mathcal{C}$ of smooth quasi-periodic vector fields (on $\R^d$ or on some smooth manifold) with vanishing average. 
 \end{enumerate}

\section{Examples of application}
\label{sec.6}

{
The nonlinear Magnus expansion has been applied in the context of control theory, namely in non-holonomic motion  planning.
The considered systems can be described by equations
\[
  q'(t) = A(t)(q) = \sum_{i=1}^m g_i(q) u_i,
\]
where $\dim q > \dim u$, $g_i$ are vector fields and the $u_i$ are the controls. The 
(nonlinear) Magnus expansion allows one to express, locally around a given point in the state space, admissible directions of
motions in terms of control parameters, so that the motion in a desired direction can be reformulated as the optimization of
those control parameters that steer the non-holonomic system into the desired direction 
\cite{strichartz87tcb,duleba97lom,duleba06pcf}. In this sense, expression (\ref{me.com3})
and the corresponding one for the generic, nonlinear case, could be very useful in applications, since it only contains  
independent terms \cite{duleba98oac,duleba06pcf}.
}

As mentioned previously, the general Floquet--Magnus can also be applied in averaging. 
A large class of problems where averaging techniques are successfully applied is made of autonomous systems of the form 
\begin{eqnarray} \label{eq:autper}
\dot u = A u + \eps h(u), \quad u(0)=x_0,
\end{eqnarray}
where $h$ is a smooth function from $ \R^n$ to itself (or more generally from a functional Banach space $E$ space to itself) and  $A$ is a skew-symmetric matrix ${\cal M}(\R^n)$ (or more generally a linear operator on $E$)  whose spectrum is included in $\frac{2i \pi}{T} \ZZ$. Denoting  $x = e^{-t A} u$ and differentiating leads to 
\begin{eqnarray*}
\dot x = -A e^{-tA} u+ e^{-tA} \dot u
		= \eps e^{-tA} h\left( e^{tA} x \right)
\end{eqnarray*}
so that $x$ now satisfies an equation of the very form (\ref{eq:odeT}) with 
$$
g(x,t) = e^{-tA} h\left( e^{tA} x \right).
$$ 
The $T$-periodicity of $g$ with respect to time stems from the fact that $\mbox{Sp}(A) \subset \frac{2 i \pi}{T} \ZZ$. For this specific case,  relation (\ref{eq.recu}) leads to the following expressions
\begin{align}
  G_1(X) =& \left\LL g(X,\cdot)  \right\RR, \\
  G_2(X)  =&  -\frac{1}{2} \left\LL \int_0^{t} \big[ g(X,\tau), g(X,t) \big] \, d\tau
  \right\RR, \label{eq:G2} \\
  G_3(X)  =& \frac{1}{12} \left\LL \int_0^{t} \left[g(X,\tau),\int_0^{t} \left[g(X,\sigma),g(X,t)\right]\, d\sigma \right]\, d\tau   \right\RR \nonumber \\
  & + \frac{1}{4} \left\LL \int_0^{t} \left[\int_0^{\tau} \left[g(X,\sigma),g(X,\tau))\right]\, d\sigma,g(X,t)\right]\, d\tau\right\RR.
\end{align}
If $g(x,t)$ is a Hamiltonian vector field with Hamiltonian $H(x,t)$, then all $G_i$'s are Hamiltonian with Hamiltonian $H_i$'s. These Hamiltonians can be computed through the same formulas with Poisson brackets {\em in lieu} of  Lie brackets (see e.g.
\cite{blanes16aci}).
 
\subsection{Dynamics of the Van der Pol system}
As a first and elementary application of previous results, we consider the Van der Pol oscillator, which may be looked at as a  perturbation of the simple harmonic  oscillator:
\begin{eqnarray}\label{eq:vdp}
\qquad &\qquad & \left\{
\begin{array}{rcl} 
\dot q &=& p \\
\dot p &=& -q + \eps (1-q^2) p
\end{array}
\right..
\end{eqnarray}
Clearly, the system is of the form (\ref{eq:autper})
with $u=(q,p)^T$, 
$$
A= \left(
\begin{array}{cc}
0 & 1\\
-1 & 0
\end{array}
\right) \quad \mbox{ and } \quad 
h(u) = \left(
\begin{array}{c}
0 \\
(1-u_1^2)u_2
\end{array}
\right),
$$
and is thus amenable to the  rewriting (\ref{eq:odeT}), where 
$$
g(x,t) := e^{-t A} h(e^{t A} x) \quad \mbox{ and }  \quad e^{t A} =  \left(
\begin{array}{cc}
\cos(t) & \sin(t) \\
-\sin(t) & \cos(t)
\end{array}
\right).
$$
In short, we have 
$$
\dot x = \eps \, \xi_t(x) V_t,
$$
where
$$
V_t = \left(
\begin{array}{r}
-\sin(t)  \\
\cos(t) 
\end{array}
\right) \quad \mbox{ and } \quad \xi_t(x) = \big(1-(\cos(t) x_1 +\sin(t) x_2)^2\Big) (-\sin(t) x_1 + \cos(t) x_2).
$$

\begin{figure}[h]
\begin{center}
\includegraphics[width=63mm,height=63mm]{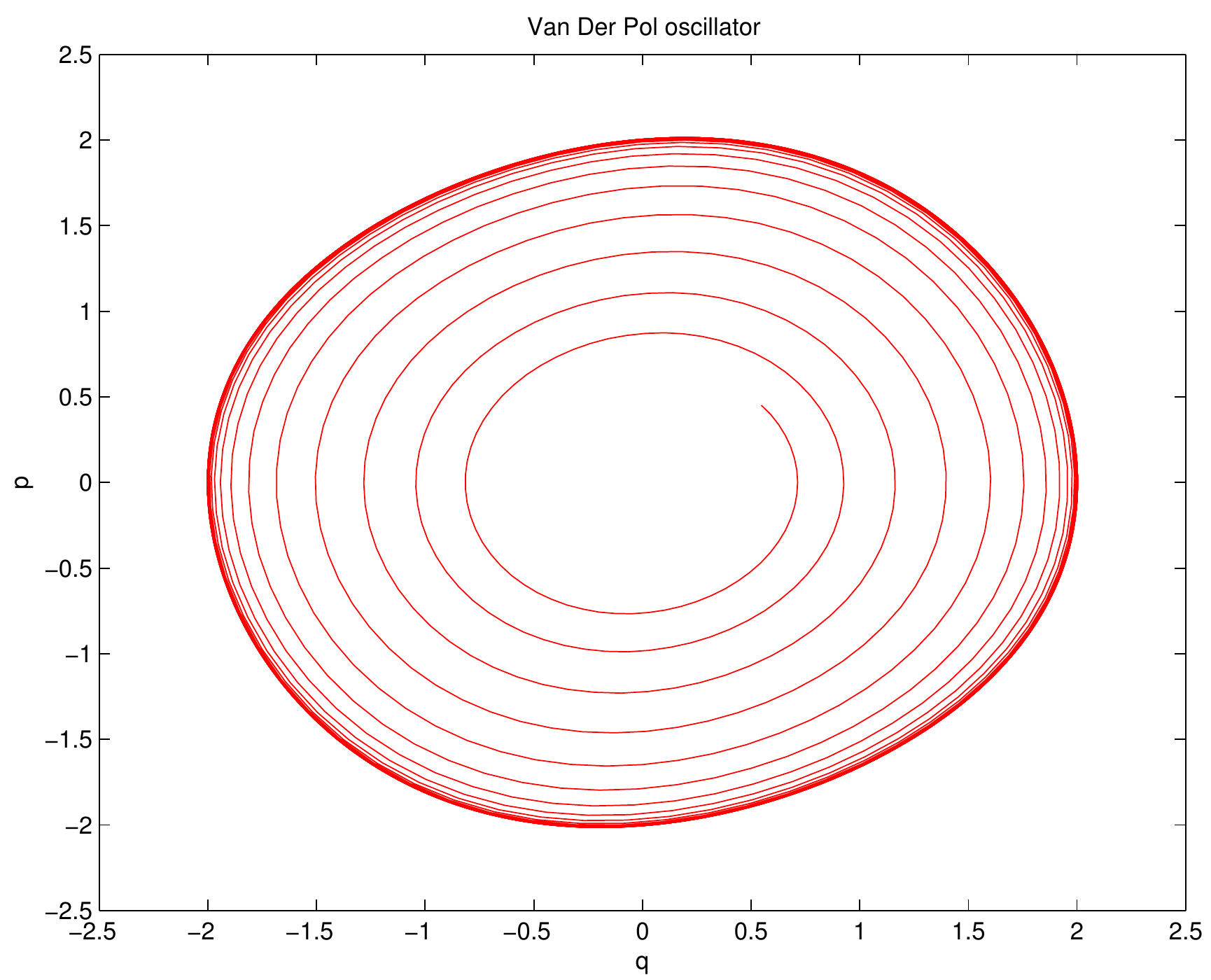} $\quad$
\includegraphics[width=63mm,height=63mm]{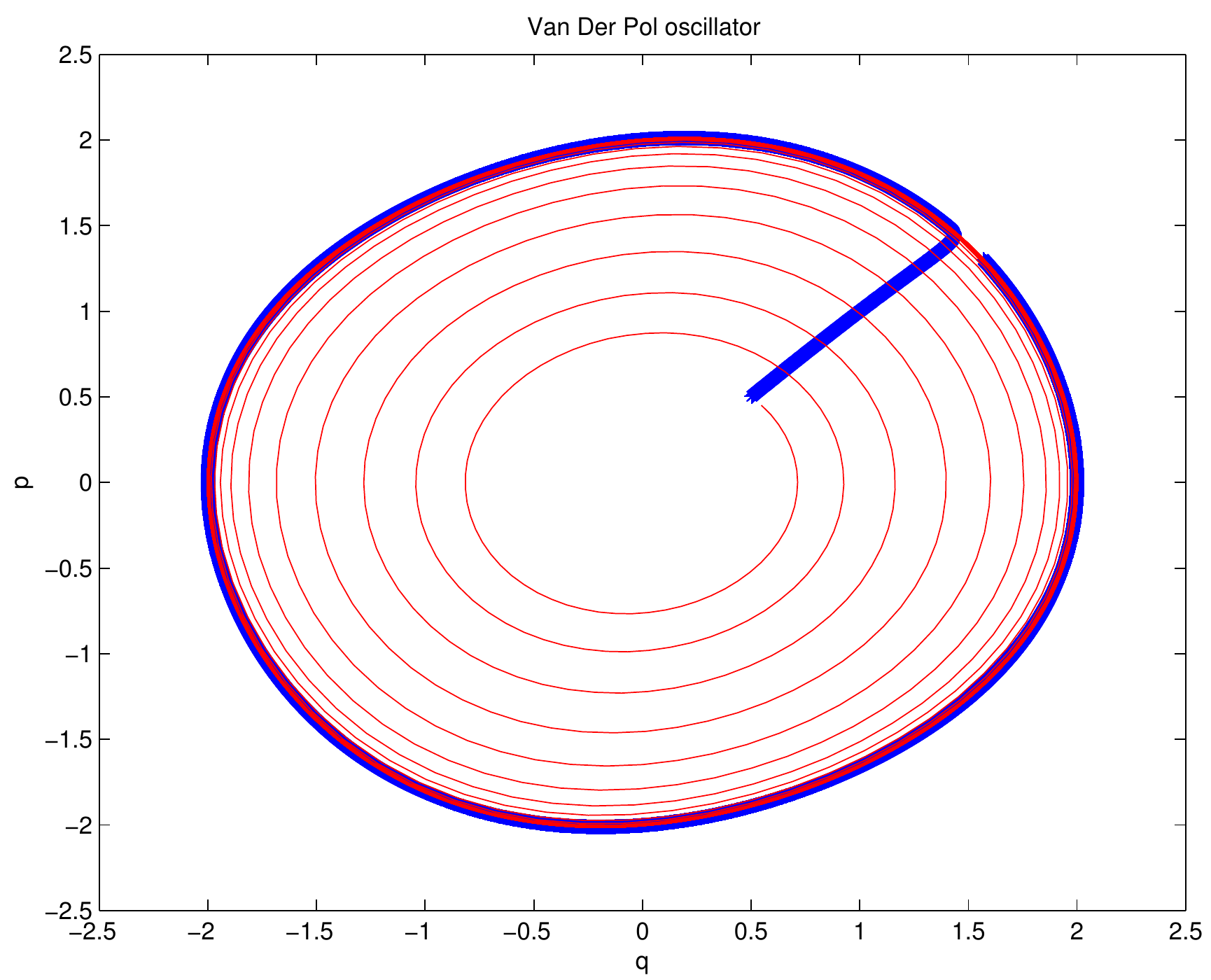}
\caption{Trajectories of the original Van Der Pol system  (in red) and of its averaged version up to second order
(in blue) \label{fig:vdp}}
\end{center}
\end{figure}
Previous formulas give for the first term in the procedure
\begin{align*}
G_1(X) =&\frac{1}{2 \pi} \int_0^{2 \pi} \xi_\tau(X) V_\tau d\tau 
	   =
\left(
\begin{array}{c}
-\frac18 (\|X\|^2_2 - 4) X_1 \\
 -\frac18 (\|X\|^2_2 -4) X_2
 \end{array}
 \right) = -\frac{(\|X\|_2^2-4)}{8} X
\end{align*}
and it is then easy to determine an approximate equation of the limit cycle, i.e.  $\|X\|_2=2$. As for the dynamics of the first-order averaged system, it is essentially governed by the scalar differential equation on $N(X):=\|X\|_2^2$
$$
\frac{d}{dt} N(X) = 2 (X_1 \dot X_1 + X_2 \dot X_2) = -\eps \frac{N (N-4)}{4},
$$
which has two equilibria, namely $\|X\|_2=0$ and $\|X\|_2=2$. The first one is unstable while the second one is stable. However, the graphical representation  (see Figure \ref{fig:vdp}) of the solution of (\ref{eq:vdp}) soon convinces us that the true limit cycle is not a perfect circle. In order to determine a better approximation of this limit cycle, we thus  compute the next term of the average equation from formula (\ref{eq:G2}):
\begin{align*}
G_2(X) =& \frac{-1}{4 \pi} \int_0^{2 \pi} \int_0^t \left( \xi_t (\nabla_X \xi_\tau, V_t) \, V_\tau-\xi_\tau (\nabla_X \xi_t, V_\tau) \, V_t \right) d\tau \\
             =& \left[ \begin {array}{c} -{\frac {1}{256}}\,X_{{2}} \left( 32-24\,{X_
{{2}}}^{2}+5\,{X_{{2}}}^{4}-88\,{X_{{1}}}^{2}+21\,{X_{{1}}}^{4}+10\,{X
_{{1}}}^{2}{X_{{2}}}^{2} \right) \\ \noalign{\medskip}{\frac {1}{256}}
\,X_{{1}} \left( 21\,{X_{{1}}}^{4}+32-88\,{X_{{1}}}^{2}+40\,{X_{{2}}}^
{2}+10\,{X_{{1}}}^{2}{X_{{2}}}^{2}+5\,{X_{{2}}}^{4} \right) 
\end {array} \right] \\
=& - \frac{\eps^2}{256} D(X) J X, 
\end{align*}
where 
$$
D(X) = \left(
\begin{array}{cc}
32-24\,{X_
{{2}}}^{2}+5\,{X_{{2}}}^{4}-88\,{X_{{1}}}^{2}+21\,{X_{{1}}}^{4}+10\,{X
_{{1}}}^{2}{X_{{2}}}^{2} & 0 \\
0 & D_{1,1}(X)+64 X_1^2
\end{array}
\right)
$$
and 
\begin{equation}  \label{defJ}
   J = \left( \begin{array}{rr}
           0  &  1  \\
           -1  &  0
        \end{array}  \right).
\end{equation}           
Now, considering the new quantity $L(X)=N(X)+\eps Q(X)$ with  
$$
Q(X)=\nu X_1 X_2^3,
$$ 
we see from the second-order averaged equation 
$$
\dot X = -\eps \frac{(\|X\|_2^2-4)}{8} X - \frac{\eps^2}{256} D(X) J X, 
$$ 
that
\begin{align*}
\frac{d L}{dt}  &= \frac{d  N}{dt} + \eps (\nabla_X Q, \dot X) \\
& = -\eps \frac{N (N-4)}{4}  -\frac{\eps^2}{128} (X,D(X) J X) - \eps^2 \frac{N-4}{8} (\nabla_X Q,  X) \\
&=  -\eps \frac{L (L-4)}{4} - \eps^2 Q +\frac{\eps^2}{2} Q L + \eps^2\frac{1}{2 \nu} Q - \eps^2\frac12(L-4) Q + {\cal O}(\eps^3) \\
&=  -\eps \frac{L (L-4)}{4} + {\cal O}(\eps^3)
\end{align*}
for $\nu=-\frac12$. A more accurate description of the limit cycle is thus given by the equation 
$$\|X\|_2^2 = 4 +\frac{\eps}{2} X_1 X_2^3.
$$

\subsection{The first two terms of the averaged  nonlinear Schr\"odinger equation (NLS) on the $d$-dimensional torus}
Next, we apply the results of Section \ref{sec.5} to the nonlinear Schr\"odinger equation (for a introduction to the NLS see for instance \cite{carles08sca,cazenave98ait,cazenave03sse})
\begin{align*}
i \pa_t \psi &= -\Delta \psi +\eps k\left(\psi \cdot \bar \psi  \right)\psi,\quad t\geq 0,\quad z \in \T_a^d,
\\
\nonumber
\psi(0,z)&=\psi_0(z)\in E,
\end{align*}
where $\T_a^d=[0,a]^d$ and $E=H^s(\T_a^d)$ is our working space.
We hereafter conform to the hypotheses of \cite{castella15saf} and assume that $h$ is a real-analytic function and that $s>d/2$, ensuring that $E$ is an algebra\footnote{Under all these assumptions, for all initial value $\psi_0 \in E$ and all $\eps >0$, there exits a unique solution $\psi \in C([0,b/\eps[,E)$ for some $b>0$ independent of $\eps$ \cite{castella15saf}.}. 
The operator $-\Delta$ is  self-adjoint non-negative and its spectrum is
\begin{align} \label{eq:spectrum}
\sigma(A)=\left\{\left(\frac{2\pi}{a}\right)^2\ \sum_{j=1}^d l_j^2 \, ; \, l \in \ZZ^d\right\}\subset \left(\frac{2\pi}{a}\right)^2 \NN,
\end{align}
so that by Stone's theorem, the group-operator $\exp(i t \Delta)$ is  periodic with period $T=\frac{a^2}{2 \pi}$.  We may thus rewrite Schrödinger equation as we did with equation (\ref{eq:autper}). However, we shall instead decompose $\psi(t,z)=q(t,z)+ip(t,z)$ in its real and imaginary parts and derive the corresponding canonical system in the new unknown
$$
u(t,\cdot)=\left(\begin{array}{c}p(t,\cdot)\\q(t,\cdot)\end{array}\right) \in H^s(\T^d_a) \times H^s(\T^d_a), 
$$
 that is to say
\begin{align}
\label{eq:y}
\dot u 
=J^{-1} \mbox{Diag}(-\Delta,-\Delta) u+\eps k\left(\|u\|_{\R^2}^2 \right)J^{-1} u,\quad
u(0)=\left(\begin{array}{c} p_0 \\ q_0 \end{array}\right),
\end{align}
where we have denoted $\dot u = \partial_t u$, $\|u\|_{\R^2}^2=(u^1)^2 + (u^2)^2 = (u,u)_{\R^2}$ ($u^1$ and $u^2$ are the two components of $u$), $J$ is given by (\ref{defJ}) and
$$
D=\left(\begin{array}{cc}-\Delta & 0\\ 0 & -\Delta \end{array}\right).
$$
The operator $D$ if self-adjoint on $L^2(\T^d_a) \times L^2(\T^d_a)$ 
and an obvious computation shows
$$
e^{t J^{-1}D}=\left(\begin{array}{cc}\cos(t \Delta)&\sin(t \Delta)\\ -\sin(t \Delta)&\cos(t \Delta)\end{array}\right) := R_t,
$$
so that $e^{t J^{-1}AD}$ is a group of isometries on $L^2(\T^d_a) \times L^2(\T^d_a)$ as well as on $H^s(\T^d_a) \times H^s(\T^d_a)$. Owing to (\ref{eq:spectrum}) it is furthermore periodic (for all $t$,  $R_{t+T}=R_t$), with period $T=\frac{a^2}{2\pi}$.
The very same manipulation as for the prototypical system (\ref{eq:autper}) then leads to 
\begin{align*}
\dot x =\eps g\left(x,t\right),\qquad x=\left(\begin{array}{c} p_0 \\ q_0 \end{array}\right),
\end{align*}
with
\begin{align*}
g(x,t)&:=J^{-1}e^{-t J^{-1} D} k \left( \|e^{t J^{-1}D}x\|_{\R^2}^2\right)e^{t J^{-1}D}x 
= R_{t+T/4} \; k \left( \|R_t x\|_{\R^2}^2\right) \; R_t x.
\end{align*}
Now, it can be verified that 
\begin{align*}
g(x,t)=J^{-1}\nabla_x  H(x,t),
\end{align*}
where 
\begin{align} \label{eq:ham}
H(x,t):=\frac{1}{2}\int_{\T_a^d} K\left(\|(R_t \, x)(t,z)\|_{\R^2}^2\right) dz \quad \mbox{ with } \quad K(r)=\displaystyle \int_0^r k(\sigma)d\sigma.
\end{align}
\begin{remark} \label{rem:ham}
Recall that the gradient is defined w.r.t. the scalar product $(\cdot,\cdot)$ on $L^2(\T^d_a) \times L^2(\T^d_a) $ that we redefine for the convenience of the reader: for all pair of functions $x_1$ and $x_2$ in $L^2(\T^d_a) \times L^2(\T^d_a) $, 
$$
(x_1,x_2) = \int_{\T_a^d} \left(x_1^1(z) x_2^1(z) +x_1^2(z) x_2^2(z) \right) dz = 
\int_{\T_a^d} (x_1(z),x_2(x))_{\R^2} dz.
$$
where $x^1_1$ and $x_1^2$ are the two components of $x_1$ and similarly for $x_2$. Hence, by definition of the gradient, we have that 
$$\forall  (t,x_1,x_2) \in \T\times E^3,\qquad  (\nabla_x H(x_1,t) , x_2)
=\pa_x H(x_1,t) \, x_2.$$
Furthermore, 
$$\forall  (t,x_1,x_2) \in \T\times E^3,\qquad  (R_t \, x_1 , x_2) = (x_1 , R_{-t} \,  x_2)$$
and 
$$\forall (x_1,x_2,x_2)\in E^3\qquad (J\partial_x g(x_1,t)x_2, x_3)=(x_2,J\partial_xg(x_1,t)x_3).$$
Finally, if $\phi_1$ and $\phi_2$ are hamiltonian vector fields, with hamiltonians $\Phi_1$ and $\Phi_2$, then 
$$\left[\phi_1,\phi_2\right]=\pa_X \phi_1 \, \phi_2 -\pa_X \phi_2  \, \phi_1 = J^{-1} \nabla_X \{\Phi_1,\Phi_2\}$$
where the Poisson bracket is defined by 
$$\left\{\Phi_1,\Phi_2\right\}=\left(J\phi_1, \phi_2\right).$$
\end{remark}
Now, the first term of the averaged vector field $G(X,\eps)$ is simply 
$$
G_1(X) = \Big\LL R_{\cdot+T/4} k\left(\|R_\cdot \, X \|^2_{\R^2} \right) \, R_\cdot \, X \Big \RR.
$$
In order to obtain the second term, we use the simple fact that for any $\delta \in H^s(\T_a^d)$ the derivatives w.r.t. $x$ in the direction $\delta$ may be computed as
\begin{align*}
\partial_x \left( k\left(\|R_t \, x\|^2_{\R^2}\right)\right) \cdot \delta &=k'\left(\|R_t \, x\|^2_{\R^2}\right) \, \left(\partial_x \|R_t \, x\|^2_{\R^2}\right) \cdot \delta  \\
&= 
2 k'\left(\|R_t \, x\|^2_{\R^2}\right) \, (R_t \, x,R_t \, \delta)_{\R^2}  
\end{align*}
so that 
\begin{align*}
\partial_x \left( g(x,t) \right) \cdot \delta =& R_{t+T/4} \, k\left(\|R_t \,  x\|^2_{\R^2}\right) R_t \, \delta \\
&+
2 R_{t+T/4} \, k'\left(\|R_t \, x\|^2_{\R^2}\right) \,  (R_t \, x,R_t \, \delta)_{\R^2}  \; R_t \,  x. 
\end{align*}
Inserted in the expression of $G_2$ we thus obtain the following expression for the $\eps^2$-term of the averaged equation
\begin{align*}
G_2(X) &= -\frac{1}{2} \left\LL \int_0^{t} \left[g(X,\tau), g(X,t)\right]\, d\tau \right\RR =I_1(X)+I_2(X)
\end{align*}
with 
\begin{align*}
I_1 (X)=& -\frac{1}{2} \left\LL \int_0^{t} \, R_{\tau +T/4} \, k\left(\|R_\tau x\|^2_{\R^2}\right) R_{\tau+t+T/4}   \, k\left(\|R_t x\|^2_{\R^2}\right) R_t \, x   d\tau \right\RR \\
&+\frac{1}{2} \left\LL \int_0^{t} \, R_{t +T/4} \, k\left(\|R_t x\|^2_{\R^2}\right) R_{\tau+t+T/4}   \, k\left(\|R_\tau x\|^2_{\R^2}\right) R_\tau \, x   d\tau \right\RR 
\end{align*}
and 
\begin{align*}
I_2 (X)=& - \left\LL \int_0^{t} \, R_{\tau +T/4} \, k'\left(\|R_\tau x\|^2_{\R^2}\right) \,  (R_\tau \, x,R_{\tau+t+T/4} k\left(\|R_t x\|^2_{\R^2}\right) R_t x)_{\R^2}  \; R_\tau \, x   d\tau \right\RR \\
&+\left\LL \int_0^{t} \, R_{t +T/4} \, k'\left(\|R_t x\|^2_{\R^2}\right) \,  (R_t \, x,R_{\tau+t+T/4} k\left(\|R_\tau x\|^2_{\R^2}\right) R_\tau x)_{\R^2}  \; R_t \, x   d\tau \right\RR.
\end{align*}
As already mentioned, both $G_1$ and $G_2$ are Hamiltonian with Hamiltonian $H_1$ and $H_2$ which could have been equivalently computed from $H(x,t)$ in (\ref{eq:ham}) (see Remark \ref{rem:ham}).


\section*{Acknowledgements}
The work of FC and AM has been supported by Ministerio de Econom{\'i}a y Com\-pe\-ti\-ti\-vi\-dad (Spain) through project MTM2016-77660-P (AEI/FEDER, UE). AM is also supported by the consolidated group IT1294-19 of the Basque Government. 
PC acknowledges funding by INRIA through its Sabbatical program and thanks the University of the Basque Country for its hospitality.

\bibliographystyle{siam}

\end{document}